\ifdefined\XeTeXversion\else
  \ifdefined\pdfoutput
    \pdfoutput=1
  \fi
\fi
\documentclass[11pt,a4paper]{article}

\usepackage[T1]{fontenc}
\usepackage[utf8]{inputenc}
\usepackage[UKenglish]{babel}
\usepackage{lmodern}
\usepackage[margin=1in]{geometry}
\usepackage{amsmath,amssymb,amsthm,mathtools}
\usepackage{booktabs,flafter}
\usepackage{microtype}
\usepackage{xcolor,tikz}
\usetikzlibrary{arrows.meta,calc,decorations.pathreplacing,fit,positioning}
\usepackage[hidelinks]{hyperref}
\hypersetup{
  pdftitle={Tight Approximation Thresholds for Geometric Stabbing},
  pdfauthor={Khaled Elbassioni, Rishikesh Gajjala, Saurabh Ray},
  pdfsubject={Geometric stabbing and approximation hardness},
  pdfkeywords={square stabbing, cube stabbing, interval stabbing,
    separated d-intervals, box stabbing, halfspace hitting,
    Unique Games, strict CSP, integrality gap}
}
\usepackage[capitalise,noabbrev]{cleveref}

\theoremstyle{plain}
\newtheorem{theorem}{Theorem}
\newtheorem{lemma}[theorem]{Lemma}
\newtheorem{corollary}[theorem]{Corollary}
\newtheorem{proposition}[theorem]{Proposition}

\AddToHook{env/theorem/begin}{\crefalias{theorem}{theorem}}
\AddToHook{env/lemma/begin}{\crefalias{theorem}{lemma}}
\AddToHook{env/corollary/begin}{\crefalias{theorem}{corollary}}
\AddToHook{env/proposition/begin}{\crefalias{theorem}{proposition}}

\title{Tight Approximation Thresholds for Geometric Stabbing}
\author{%
  Khaled Elbassioni\textsuperscript{1}\quad
  Rishikesh Gajjala\textsuperscript{2}\quad
  Saurabh Ray\textsuperscript{2}\\[0.5em]
  \small\textsuperscript{1}Khalifa University, Abu Dhabi, United Arab Emirates\\
  \small\textsuperscript{2}New York University Abu Dhabi, United Arab Emirates
}
\date{}

\newcommand{\eps}{\varepsilon}
\newcommand{\R}{\mathbb R}
\newcommand{\E}{\mathbb E}
\newcommand{\Prb}{\mathbb P}
\newcommand{\supp}{\operatorname{supp}}
\newcommand{\opt}{\operatorname{opt}}

\newcommand{\Inf}{\operatorname{Inf}}
\newcommand{\seed}{\operatorname{seed}}
\newcommand{\Top}{\mathsf{Top}}
\newcommand{\one}{\mathbf 1}
\newcommand{\cL}{\mathcal L}
\newcommand{\cR}{\mathcal R}

\newcommand{\cU}{\mathcal U}
\newcommand{\cB}{\mathcal B}

\newcommand{\pp}[1]{\par\medskip\noindent\textbf{\emph{#1}}\enspace}

\definecolor{vblue}{RGB}{34,91,155}
\definecolor{horange}{RGB}{213,94,0}
\definecolor{softblue}{RGB}{226,239,249}
\definecolor{gapgreen}{RGB}{0,135,108}
\definecolor{seedgreen}{RGB}{0,158,115}
\definecolor{seedpurple}{RGB}{117,112,179}
\definecolor{seedrose}{RGB}{204,121,167}

\begin{document}
\maketitle

\begin{abstract}
We study geometric covering problems: stabbing axis-parallel squares
and cubes by coordinate hyperplanes, stabbing horizontal segments by
horizontal and vertical lines, and transversal of separated
\(d\)-intervals.  We prove matching NP-hardness bounds for their natural covering-LP thresholds.  For every \(\eps>0\), square
stabbing is NP-hard to approximate within \(2-\eps\), horizontal-segment
stabbing within \(e/(e-1)-\eps\), and, for every fixed \(d\ge2\), cube
stabbing and separated \(d\)-interval transversal within \(d-\eps\).  The
results hold with finite candidate sets and when arbitrary admissible lines,
hyperplanes, or track points are allowed.

We construct harmonic integrality-gap families for squares and cubes and
use the Kovaleva--Spieksma gap family for horizontal segments.  In each
family, every object meets a candidate of every orientation.  We lift
these gaps through the strict-CSP framework of Kumar, Manokaran, Tulsiani,
and Vishnoi, using the Unique Games theorem to obtain NP-hardness. An ordered sampler preserves the LP marginals, coupled
shifts make every local sample feasible, and a small perturbation through
an absorbing label connects every local support.  Cartesian cloning
removes rational weights, after which disjoint blocks of candidate
coordinates realize each output constraint by one square, cube, or
horizontal segment.

Applying the same transfer theorem to rectangle-hitting gaps of Pach and
Tardos, we prove that stabbing axis-parallel boxes by supplied candidate
lines in \(\R^3\) is NP-hard to approximate within any fixed constant
factor. Incidence-preserving algebraic maps yield the same hardness for
hitting boxes in \(\R^3\) and halfspaces in \(\R^6\) by points from supplied
finite candidate sets.

\end{abstract}


\section{Introduction}
\label{sec:introduction}

For fixed \(d\ge2\), let \(S\) be a finite family of axis-parallel cubes in
\(\mathbb R^d\).  A coordinate-orthogonal hyperplane \emph{stabs} a cube if
it intersects the cube, and the cube-stabbing problem asks for a
minimum-cardinality set of such hyperplanes stabbing every member of
\(S\).  Square stabbing is precisely the case \(d=2\).
For coordinate-hyperplane stabbing of squares and cubes, the candidates
of each orientation meeting an object form a consecutive interval in
their geometric order.

The standard covering LP has one variable for every candidate hyperplane
and requires incident mass at least one on every object.  Approximation
algorithms for geometric line stabbing go back to Hassin and
Megiddo~\cite{HassinMegiddo1991}.  Gaur, Ibaraki, and
Krishnamurti~\cite{GIK2000,GIK2002} gave the classical LP-relative
factor-\(2\) approximation for rectangle stabbing, first published at
ESA~2000.  Its direct extension gives factor \(d\) for boxes: assign every
object to an orientation carrying mass at least \(1/d\) and use the
integrality of interval covering on each coordinate axis.
Combining the Unique Games theorem of
OpenAI~\cite[Theorem~1.1]{OpenAI2026UGT} with the strict-CSP framework
of Kumar, Manokaran, Tulsiani, and Vishnoi~\cite{KMTV2011,Manokaran2012},
we prove that approximating square stabbing within \(2-\eps\) is NP-hard for every fixed \(\eps>0\).

Congruent squares admit better bounds: Elbassioni and
Ray~\cite{ElbassioniRay2024} obtained an LP-relative approximation ratio
strictly below \(2\).  In our case, the squares
may have arbitrary, and in particular unequal, side lengths. We show
that allowing unequal side lengths restores the tight factor \(2\):
orientation rounding is optimal unless \(P=NP\).
This parallels Khot and Regev's factor-\(2\) hardness for Vertex
Cover~\cite{KhotRegev2008}, whose Unique Games Conjecture assumption is now supplied by the Unique Games theorem~\cite{OpenAI2026UGT}.

Another basic specialization is \emph{interval stabbing}: the objects are
horizontal segments, and each may be stabbed by its supporting horizontal
line or by a vertical line crossing it.  Kovaleva and Spieksma proved the
tight LP-relative factor \(e/(e-1)\)~\cite{KovalevaSpieksma2006}.  We lift
their gap family to matching NP-hardness.

Cube stabbing is also a special case of transversal of separated
\(d\)-intervals.  Such an object has one interval on each of \(d\) disjoint
ordered tracks, and a transversal is a set of track points meeting every
object.  Tardos~\cite{Tardos1995}, Kaiser~\cite{Kaiser1997},
Alon~\cite{Alon1998}, and Matoušek~\cite{Matousek2001} established central
packing--covering bounds for \(d\)-intervals.  Ben-David, Grant, Ma, and
Sharpe~\cite{BenDavidEtAl2012} gave a factor-\(d\) approximation and
separated instances whose natural-LP gap tends to \(d\); they call the
separated-track objects \emph{\(d\)-union-intervals}.  Our cube result
gives matching NP-hardness when all components are non-empty
and the \(d\) components of each object have equal length.
In the following, we develop the
case $d=2$ first because it displays the common construction with the simplest
notation and geometry, and then extend it to cubes in every fixed
dimension.

\paragraph*{Results.}
The \emph{finite-candidate model} supplies finite sets of admissible
stabbing lines, hyperplanes, or points.  In the \emph{unrestricted
model}, any stabbing object of the appropriate type may be chosen.

\begin{theorem}[Cube stabbing]
\label{thm:main-cubes}
For every fixed \(d\ge2\) and every \(\eps>0\), unit-cost
stabbing of axis-parallel cubes in \(\mathbb R^d\) by coordinate-orthogonal
hyperplanes is NP-hard to approximate within \(d-\eps\), in both the
finite-candidate and unrestricted models.  The factor \(d\) equals the
supremum of the integrality gap of the natural covering LP.  In particular,
for \(d=2\) this gives the tight factor \(2\) for square stabbing.
\end{theorem}

\begin{theorem}[Horizontal-segment stabbing]
\label{thm:main-segments}
For every \(\eps>0\), unit-cost stabbing of horizontal
segments by horizontal and vertical lines is NP-hard to approximate within
\(e/(e-1)-\eps\), in both the finite-candidate and unrestricted models.
The factor \(e/(e-1)\) equals the supremum of the integrality gap of the
natural covering LP.
\end{theorem}

\begin{corollary}[Separated \texorpdfstring{\(d\)}{d}-intervals]
\label{cor:main-dintervals}
For every fixed \(d\ge2\) and every \(\eps>0\), transversal of
separated \(d\)-intervals is NP-hard to approximate within \(d-\eps\), in
both the finite-candidate and unrestricted models, even when every
component is nonempty and the \(d\) components of each object have equal
length.  The factor \(d\) equals the supremum of the integrality gap of the
natural covering LP.
\end{corollary}

We also consider the problem of stabbing axis-parallel boxes by candidate lines in \(\R^3\).
Pach and Tardos~\cite{PachTardos2013} constructed planar rectangle-hitting
instances with arbitrarily large natural-LP gaps.  We transfer these gaps
through a binary seed and realize the output by boxes stabbed by candidate lines to prove the following result.

\begin{theorem}
\label{thm:box-lines}
For every fixed \(C>1\), unit-cost stabbing of axis-parallel
boxes in \(\R^3\) using a supplied finite set of axis-parallel lines
is NP-hard to approximate within \(C\).  This holds when the candidates
have only the \(x\)- and \(y\)-directions and every box has a square
\(xy\)-projection.
\end{theorem}

Incidence-preserving algebraic maps then give the following results for
hitting boxes and halfspaces by points from supplied finite sets.

\begin{corollary}[Box hitting]
\label{cor:main-box-points}
For every fixed \(C>1\), it is NP-hard to approximate within
\(C\) the minimum-cardinality subset of a supplied finite point set in
\(\R^3\) hitting every member of a given family of axis-parallel boxes.
This holds even when all candidate points lie on two parallel planes.
\end{corollary}

\begin{corollary}[Halfspace hitting]
\label{thm:main-halfspaces}
For every fixed \(C>1\), it is NP-hard to approximate within
\(C\) the minimum-cardinality subset of a supplied finite point set in
\(\R^6\) hitting every member of a given family of halfspaces.
\end{corollary}

The latter three results concern the finite-candidate model.
For candidate-point hitting of boxes in \(\R^3\), Aronov, Ezra, and
Sharir~\cite{AronovEzraSharir2010} gave an
\(O(\log\log\opt)\)-approximation in randomized expected polynomial time.
Chan and Grant~\cite{ChanGrant2014} proved APX-hardness for
candidate-point hitting of halfspaces in \(\R^4\).

\pp{Proof overview.}
For squares, cubes, and horizontal segments, we encode a geometric gap
instance as a finite strict CSP with one seed variable per orientation.
A label specifies an entire subset of that
orientation's candidates.  An ordered sampler and constraint-specific
couplings turn a fractional cover into consistent local distributions;
a small absorbing-label perturbation connects their supports.
The KMTV framework~\cite{KMTV2011,Manokaran2012} lifts the seed gap while
preserving each constraint's ordered traces.  After removing weights by
Cartesian cloning, we give each variable a private candidate block and
translate the inherited traces into those blocks.  Equal projection
lengths preserve squares and cubes, while a singleton horizontal trace
preserves horizontal segments.  The geometric realization is exact;
perturbation, transfer, and cloning incur an arbitrarily small loss.

For boxes and halfspaces, the rectangle \(\eps\)-net lower bound gives a
rectangle-hitting instance with fractional cost \(1/\eps\) and integral
cost \(\Omega(\eps^{-1}\log\log(1/\eps))\).  Thus its LP gap is
arbitrarily large.  We encode it by two subset-valued seed variables;
product sampling and a coupling of hit indicators give connected local
supports.  Transfer and cloning yield constraints realized by
axis-parallel boxes stabbed by candidate lines in \(\R^3\).
Incidence-preserving algebraic maps then give candidate-point hitting
instances for boxes in \(\R^3\) and halfspaces in \(\R^6\), preserving
the cost of every candidate subset.

\pp{Organization.}
\Cref{sec:preliminaries} gives the geometric and CSP preliminaries.
\Cref{sec:square-gap} proves the square LP gap, \cref{sec:kmtv} states the
transfer theorem, and \cref{sec:square-hardness} proves square hardness.
\Cref{sec:cubes} sketches the extensions to cubes, separated
\(d\)-intervals, and horizontal segments.
\Cref{sec:halfspaces} proves NP-hardness within every fixed constant factor for candidate-line
box stabbing and candidate-point hitting of boxes and halfspaces.
Appendices~\ref{app:kmtv}--\ref{app:segments} give the transfer proof and
the technical details for cubes and segments.

\section{Preliminaries}
\label{sec:preliminaries}

\subsection{Geometric stabbing}

For \(N\in\mathbb N\), write \([N]=\{1,\ldots,N\}\).  A coordinate
hyperplane of type \(i\in[d]\) is
\(\pi_{i,\xi}=\{x\in\R^d:x_i=\xi\}\).  In the plane, types \(1\) and
\(2\) are called vertical and horizontal, in this order.

All objects are closed and, except for the halfspaces in
\cref{sec:halfspaces}, bounded; input coordinates and defining
coefficients are rational.
In the finite-candidate model, every object is assumed to meet a
candidate; integrality gaps are taken over nonempty feasible instances.
Let \(\mathcal P\) be a finite candidate set and \(S\) the object family.
The natural covering LP is
\begin{equation}
 \min\left\{\sum_{\pi\in\mathcal P}x_\pi:
 \sum_{\pi:\pi\cap s\ne\varnothing}x_\pi\ge1\ (s\in S),\quad
 0\le x_\pi\le1\ (\pi\in\mathcal P)\right\}.
\label{lp:cover}
\end{equation}
For a fixed coordinate \(i\), suppose the type-\(i\) candidates are indexed
in geometric order.  The \emph{type-\(i\) trace} of an object is the set of
indices of type-\(i\) candidates meeting it.  For coordinate-hyperplane
stabbing, every type-\(i\) trace of an axis-parallel box is consecutive.
For a candidate-weight vector \(x\)
and a trace \(T\), write \(x(T)=\sum_{j\in T}x_j\).

For coordinate-hyperplane stabbing in the unrestricted model, it suffices
to consider hyperplanes through
projection endpoints.  Indeed, the largest lower endpoint among the
projections met by a useful hyperplane belongs to every one of them.
The analogous candidates are component endpoints on each track for
separated intervals, and supporting rows and vertical lines through
segment endpoints for horizontal segments.  This gives \(O(d|S|)\)
candidates for boxes and separated intervals, and \(O(|S|)\) for horizontal
segments.  Snapping integral choices or fractional mass to these
candidates preserves feasibility and does not increase cost; after
combining fractional masses, cap each total at one.  If a total is capped,
that candidate alone covers every object it meets.  This finite
discretization defines the unrestricted natural LP and preserves both
optima.  A \emph{useful} candidate is one meeting at least one object;
candidates meeting none can always be discarded.

\subsection{Unique Games}

A Unique Games instance is a graph whose vertices must be assigned labels
from a finite set \([R]\).  Fix an arbitrary ordering of the vertices.
For every edge \(xy\) with \(x<y\), the instance specifies one permutation
\(\pi_{x,y}\) of \([R]\).  A labeling \(\Lambda\) satisfies this edge
precisely when \(\Lambda(y)=\pi_{x,y}(\Lambda(x))\).  Thus, once the label
at one endpoint is fixed, exactly one label at the other endpoint satisfies
the edge.  Edges may carry nonnegative rational weights summing to one;
the value is the maximum total weight of satisfied edges.
The Unique Games theorem~\cite[Theorem~1.1]{OpenAI2026UGT}
establishes the promise originally conjectured by Khot~\cite{Khot2002}:
for every fixed \(0<\eta,\zeta<1/2\), some fixed alphabet size
\(R=R(\eta,\zeta)\) makes it NP-hard to distinguish instances of
value at least \(1-\eta\) from instances of value at most \(\zeta\).

\subsection{Strict CSPs}

Let \(\Sigma\) be a finite alphabet.  A finite strict-CSP language is an
arity-indexed family
\[
 \cL=(\cL_1,\cL_2,\ldots),\qquad \cL_r\subseteq 2^{\Sigma^r},
\]
with only finitely many nonempty \(\cL_r\).  Thus \(\cL_r\) is a set of
allowed \(r\)-ary rules; each rule is a subset of
\(\Sigma^r\).

An instance over \(\cL\) consists of a variable set \(V\), a rational
nonnegative weight \(w_v\) on each variable \(v\in V\), with
\(\sum_{v\in V}w_v=1\), a bounded rational cost
\(c:\Sigma\to\R_{\ge0}\), and a finite set \(C\) of constraints.
A constraint \(f\in C\) is defined by a rule \(\cR_f\in\cL_r\) and an
ordered scope \((v_1,\ldots,v_r)\in V^r\).  An assignment
\(\sigma:V\to\Sigma\) satisfies \(f\) if
\((\sigma(v_1),\ldots,\sigma(v_r))\in\cR_f\).  Constraints have distinct
identifiers, even if their rules and scopes coincide.  Such identical
copies may be removed without changing feasibility or optimum; retaining
them merely records their provenance.  Constraints on the same scope
with different rules are retained.  An assignment is feasible if it
satisfies every constraint.  Unlike in a Max-CSP, no constraint may be violated.  Its cost is
\(\sum_vw_vc(\sigma(v))\).

A label \(\Top\in\Sigma\) is \emph{absorbing} if, for every \(r\) and
every rule \(\cR\in\cL_r\), each tuple in \(\Sigma^r\) containing
\(\Top\) in at least one coordinate belongs to \(\cR\).

The \emph{basic LP relaxation} chooses a distribution \(\mu_v\) on
\(\Sigma\) for each variable and a distribution \(\lambda_f\) on the
satisfying tuples of each constraint.  For later use, \(\delta_a\) denotes
the point-mass distribution at \(a\in\Sigma\).  The LP is
\[
\begin{array}{ll}
\text{minimize}
  & \displaystyle\sum_{v\in V}w_v\sum_{a\in\Sigma}c(a)\mu_v(a),\\[2mm]
\text{subject to}
  & \displaystyle\sum_{a\in\Sigma}\mu_v(a)=1
      \qquad (v\in V),\\[2mm]
  & \displaystyle\sum_{\boldsymbol a\in\cR_f}\lambda_f(\boldsymbol a)=1
      \qquad (f\in C),\\[2mm]
  & \displaystyle\sum_{\substack{\boldsymbol a\in\cR_f\\a_i=a}}
       \lambda_f(\boldsymbol a)=\mu_{v_i}(a)
      \qquad \bigl(f\in C,\ i\in[r],\ a\in\Sigma\bigr),\\[2mm]
  & \mu_v(a),\lambda_f(\boldsymbol a)\ge0.
\end{array}
\]
In the third constraint, \((v_1,\ldots,v_r)\) denotes the ordered scope of
the particular constraint \(f\) under consideration.  The coefficients
\(\lambda_f\) form a distribution on the satisfying tuples of \(f\), and
the equality says that its \(i\)th marginal equals \(\mu_{v_i}\).  It is
called \emph{basic} because it imposes only variable-to-constraint
consistency.

The \emph{support} of \(\lambda_f\) is
\(\supp(\lambda_f)=\{\boldsymbol a\in\cR_f:\lambda_f(\boldsymbol a)>0\}\),
the set of satisfying tuples that receive positive probability.  Construct a
graph on this support by joining two tuples when they differ in exactly one CSP
coordinate.  We call the support \emph{connected} if this graph is connected.

\section{Integrality Gap for Square Stabbing}
\label{sec:square-gap}

We first establish that the natural covering LP can always be rounded
within factor \(2\) for squares.  We then construct a family for which
the ratio between the integral and fractional optima approaches \(2\),
showing that this bound is tight.
The planar rectangle-stabbing setting was studied by Gaur,
Ibaraki, and Krishnamurti~\cite{GIK2002}; the following direct argument
records the dimension-\(d\) orientation-rounding form needed here.

\begin{proposition}[Orientation rounding]
\label{prop:orientation-rounding}
For every fixed \(d\ge2\), \cref{lp:cover} can be rounded within factor
\(d\) for axis-parallel \(d\)-cubes.
\end{proposition}

\begin{proof}
Assign each cube to an orientation carrying mass at least \(1/d\).
For candidates in that orientation use \(y_\pi=\min\{1,dx_\pi\}\).
Each assigned cube has mass at least one: if a candidate meeting it is
capped, that candidate alone covers it; otherwise its mass is the
uncapped scaled mass.  Within one orientation the problem is interval
covering on one axis.  Its incidence matrix has the consecutive-ones
property, so its covering LP is integral.  Combining these integral
covers costs at most \(\sum_\pi y_\pi\le d\sum_\pi x_\pi\).
\end{proof}

To match the factor \(2\) upper bound for squares in the plane, fix
integers \(t,N\ge1\).  Place
vertical and horizontal candidate lines at
coordinates \(1,\ldots,N\).  For
\(1\le w\le N\) and \(1\le a,b\le N-w+1\), let
\[
 Q(a,b,w)=[a-\tfrac14,a+w-\tfrac34]
            \times[b-\tfrac14,b+w-\tfrac34].
\]
The family \(S_{t,N}\) consists of the squares \(Q(a,b,w)\) satisfying
\begin{equation}
                         a+b-1\le tw.
\label{eq:square-retain}
\end{equation}
The vertical and horizontal traces of the square \(s=Q(a,b,w)\) are
\(I_s=\{a,\ldots,a+w-1\}\) and
\(J_s=\{b,\ldots,b+w-1\}\), respectively.  In particular, every square is hit by at
least one candidate vertical line and at least one candidate horizontal
line.  This property is essential in \cref{sec:square-hardness}.

Write \(H_N=\sum_{j=1}^N1/j\) for the \(N\)th harmonic number and
\(K_{t,N}\) for the integral optimum of \(S_{t,N}\).

\begin{theorem}[Square gap]
\label{thm:square-gap}
The instance \(S_{t,N}\) has a fractional cover of cost at most
\(L_{t,N}=\frac{t+1}{2}H_N\), whereas its integral optimum satisfies
\begin{equation}
 K_{t,N}\ge t\ln\left(1+\frac{N+1}{t+1}\right)-1.
\label{eq:square-integral}
\end{equation}
Since the fractional optimum is at most \(L_{t,N}\),
\(K_{t,N}/L_{t,N}\) is a lower bound on the LP gap.  Moreover,
\(\liminf_{N\to\infty}K_{t,N}/L_{t,N}\ge 2t/(t+1)\), so the
integrality-gap supremum is \(2\).  The same holds in the unrestricted model.
\end{theorem}

\begin{figure}[t]
\centering
\begin{tikzpicture}[x=.55cm,y=.55cm,font=\small]
  \foreach \j in {1,...,8}{
    \pgfmathsetmacro{\lineweight}{2.00-0.22*(\j-1)}
    \draw[vblue!75,line width=\lineweight pt] (\j,.55)--(\j,8.45);
    \draw[horange!78,line width=\lineweight pt] (.55,\j)--(8.45,\j);
    \node[below=2pt] at (\j,.55) {\(\j\)};
    \node[left=2pt] at (.55,\j) {\(\j\)};
  }
  \filldraw[fill=vblue!70,fill opacity=.25,draw=vblue,draw opacity=1,thick]
    (.75,1.75) rectangle (1.25,2.25);
  \filldraw[fill=horange!70,fill opacity=.24,draw=horange,draw opacity=1,thick]
    (1.75,.75) rectangle (3.25,2.25);
  \filldraw[fill=gapgreen!65,fill opacity=.22,draw=gapgreen,draw opacity=1,thick]
    (.75,3.75) rectangle (3.25,6.25);
  \filldraw[fill=gray!70,fill opacity=.22,draw=gray!75,draw opacity=1,thick]
    (3.75,.75) rectangle (6.25,3.25);
  \node[below=12pt] at (4.5,.55)
    {vertical lines: coordinate \(x\)};
  \node[above=5pt] at (4.5,8.45)
    {horizontal lines: coordinate \(y\)};
\end{tikzpicture}
\caption{A portion of \(S_{2,8}\).  Line thickness suggests the harmonic
weight \(x_j\propto1/j\).  The four displayed squares have
width parameters \(1,2,3,3\); each projection contains a nonempty
consecutive set of candidates.}
\label{fig:gap-family}
\end{figure}

\begin{proof}
Give the line at coordinate \(j\), in either orientation, weight
\(x_j=\min\{1,(t+1)/(4j)\}\).
For \(s=Q(a,b,w)\), before capping the weights at \(1\),
Cauchy--Schwarz gives
\[
 \sum_{j=a}^{a+w-1}\frac{t+1}{4j}
 +\sum_{j=b}^{b+w-1}\frac{t+1}{4j}
 \ge \frac{t+1}{4}\frac{(2w)^2}{w(a+b+w-1)}\ge1,
\]
where the last inequality follows from \cref{eq:square-retain}.  If a
variable corresponding to a line meeting \(s\) is capped at \(1\), that
variable alone covers \(s\); otherwise the displayed uncapped sum is the
actual covering mass of \(s\).  Hence the capped solution is feasible.
Its cost is at most
\(2\sum_{j=1}^N(t+1)/(4j)=L_{t,N}\).

Fix an integral cover of size \(K\).  On each axis add sentinels at \(0\)
and \(N+1\).  Consecutive selected coordinates, including sentinels,
define \emph{physical gaps}; a gap of length \(h\) contains \(h-1\)
unselected candidates, meaning integer-coordinate candidate lines not
chosen by the cover.  If the cover uses \(K_x\) vertical and \(K_y\)
horizontal lines, the two axes have \(K_x+1\) and \(K_y+1\) gaps,
respectively, and hence \(K+2\) gaps in total.  Process these gaps one at a
time in nondecreasing order of length, breaking ties arbitrarily, and stop
as soon as all gaps on one axis have been processed.  If \(q\) gaps have
then been processed, at least one gap on the other axis remains, so
\(q\le K+1\).  Moreover, the processed gaps include every gap on the
exhausted axis and therefore have total length at least \(N+1\).

Immediately before processing a gap of length \(h\), let \(T\) be the sum
of the lengths of all gaps processed so far.  We claim
\begin{equation}
                              th<T+t+1.
\label{eq:square-growth}
\end{equation}
This is immediate for \(h=1\).  Assume \(h>1\).  Since the next gap is
about to be processed, the procedure has not yet stopped, so each axis
still has an unprocessed gap.  Every unprocessed gap has length at least
\(h\).  On each axis choose the physical gap of length at least \(h\)
that is closest to the origin, allowing an already processed gap tied
at length \(h\), and let \(a\), respectively \(b\), be the smallest integer
coordinate strictly inside that gap.  Thus \(a\) and
\(b\) are the closest-to-origin unselected integer coordinates in the
chosen gaps.  Every gap spatially preceding either chosen gap has length
less than \(h\) and was processed earlier.  Their combined length is
\((a-1)+(b-1)\le T\).

The \(h-1\) consecutive unselected coordinates beginning at \(a\) on the
vertical axis and at \(b\) on the horizontal axis are the two traces of
\(Q(a,b,h-1)\).  The cover misses this square.  Consequently
\(Q(a,b,h-1)\notin S_{t,N}\), since otherwise the chosen lines would not
cover \(S_{t,N}\).  The defining inequality for \(S_{t,N}\) therefore
fails, and
\[
 t(h-1)<a+b-1\le T+1,
\]
which proves \cref{eq:square-growth}.

For \(m=0,\ldots,q\), let \(T_m\) be the sum of the lengths of the first
\(m\) processed gaps and set \(Z_m=T_m+t+1\).  Thus \(Z_0=t+1\).
If the \((m+1)\)st gap has length \(h\), then
\cref{eq:square-growth} gives \(h<Z_m/t\), and hence
\[
 Z_{m+1}=Z_m+h<(1+1/t)Z_m.
\]
At termination \(T_q\ge N+1\), so \(Z_q\ge N+t+2\).  Iterating the
preceding inequality and using \(q\le K+1\) gives
\[
 N+t+2<(t+1)(1+1/t)^q\le(t+1)e^{(K+1)/t},
\]
which yields \cref{eq:square-integral}.  Since \(H_N=\ln N+O(1)\), the
ratio tends to at least \(2t/(t+1)\) as \(N\) tends to infinity.  Letting
\(t\) tend to infinity gives the lower bound \(2\).

Finally, discard lines meeting no square and map each remaining real
line to a parallel line whose coordinate is a nearest integer.  If a vertical line at coordinate \(\xi\) meets
\(Q(a,b,w)\), then
\[
 \xi\in[a-\tfrac14,a+w-\tfrac34]
       \subset(a-\tfrac12,a+w-\tfrac12).
\]
Therefore every nearest integer to \(\xi\) belongs to
\(\{a,\ldots,a+w-1\}\), so the rounded line still meets the square.  The
same argument applies to horizontal lines.  Rounding every line in an
unrestricted integral cover thus gives a candidate-line cover of no
larger size.  For a fractional cover, move the weight to the rounded candidates,
combine coincident masses, and cap each total at one.  This preserves
feasibility and does not increase cost.  Since every candidate is also
admissible in the unrestricted model, both integral and fractional optima
are equal in the two models.  The matching upper bound follows from
\cref{prop:orientation-rounding}.
\end{proof}

\section{The KMTV Transfer Theorem}
\label{sec:kmtv}

Kumar, Manokaran, Tulsiani, and Vishnoi (KMTV) provide an LP-based
hardness framework based on Unique Games for strict CSPs. We use its connected-support
cost-minimization specialization, recording the seed parent of each output
variable and constraint for geometric realization.  The transfer
framework is due to KMTV~\cite[Sections~6--7]{KMTV2011}; see also
Manokaran~\cite[Sections~5.1--5.2 and Theorem~5.4.1]{Manokaran2012}.
Section and corollary numbers for KMTV refer to the author-hosted full
version linked in the bibliography. The strong-left Unique Games source of OpenAI~\cite[Lemma~A.1]{OpenAI2026UGT} makes the resulting gap NP-hard.

\begin{theorem}[Connected-support transfer]
\label{thm:transfer}
Fix a finite strict-CSP seed with variable set \(V\) and constraint set
\(C\).  Suppose that each scope contains
distinct variables, variable weights are rational and sum to one, the
common label cost is bounded and rational, and the language has an
absorbing label \(\Top\).  Let the integral optimum be \(\kappa\).  If the
seed has a rational feasible basic-LP solution of value \(\ell<\kappa\)
whose local supports are connected, then, for every
\[
                     0<\delta<(\kappa-\ell)/2,
\]
it is NP-hard to distinguish output instances of optimum at most
\(\ell+\delta\) from instances of optimum at least
\(\kappa-\delta\). Output weights are rational.

Moreover, the construction supplies maps
\[
 \seed:V_{\rm out}\to V,\qquad \seed:C_{\rm out}\to C.
\]
If an output constraint \(g\) has scope \((x_1,\ldots,x_r)\),
\(\seed(g)=f\), and \(f\) has scope \((v_1,\ldots,v_r)\), then
\(\seed(x_i)=v_i\) for every
\(i\in[r]\), and \(g\) uses exactly \(\cR_f\), in this coordinate order.
Each generated constraint retains its seed identifier.
\end{theorem}

\pp{Construction and provenance.}
Combine the fixed seed and its LP solution with a bipartite Unique Games
instance on vertex classes \(\cU,\cB\).  For each \(u\in\cU\) and seed
variable \(v\), create a \emph{cloud} of output variables
\((u,v,z)\), \(z\in\Sigma^R\), with word weights proportional to
\(\mu_v^{\otimes R}(z)\).  Thus \(\seed(u,v,z)=v\).
The Unique Games labels select word coordinates, and its edge
permutations align those coordinates between neighboring vertices.

For each seed constraint \(f\) on \((v_1,\ldots,v_r)\) and each
\(b\in\cB\), choose one edge \((u_i,b)\) per scope position and words
whose every column lies in \(\supp(\lambda_f)\).  After applying the
corresponding coordinate permutations, impose \(\cR_f\) on these
\(r\) output variables.  Word coordinates are permuted, but scope
positions and seed parents are unchanged.  Appendix~\ref{app:kmtv}
specifies the weights, permutations, and complete constraint enumeration,
and proves \cref{thm:transfer} using the connected-space positivity lemma.

\pp{Gap preservation.}
In the YES case, the strong-left Unique Games source supplies left
vertices of total weight at least \(1-\eta\) whose incident edges are all satisfied.
Use the selected word coordinate on their clouds and \(\Top\) on the
remaining clouds.  A constraint entirely in the first group sees one
supported seed column; every other constraint contains \(\Top\).
Thus all constraints are feasible.  The objective on the first group
reproduces the seed LP value, and the exceptional clouds add at most
\(\eta c(\Top)\).

For soundness, connected-space positivity turns low-influence behavior
into a feasible seed assignment of almost the same cost.  An output
assignment cheaper than \(\kappa-\delta\) must therefore have influential
word coordinates on a positive-weight set of clouds.  Decoding these
coordinates yields a labeling satisfying a positive weight of Unique
Games edges, contradicting the NO-case promise.  Appendix~\ref{app:kmtv}
gives the complete argument and parameter choices.

\section{Hardness for Square Stabbing}
\label{sec:square-hardness}

\begin{theorem}
\label{thm:square-hardness}
For every \(\eps>0\), unit-cost stabbing of arbitrary-size
axis-parallel squares is NP-hard to approximate within \(2-\eps\), in both
the finite-candidate and unrestricted models.
\end{theorem}

\subsection{The seed CSP}

Choose \(t,N\) so that the instance \(S=S_{t,N}\) from
\cref{sec:square-gap} has a rational feasible cover \(x\) of cost
\(L'<K\), where \(K\) is the integral optimum and
\(K/L'>2-\eps/4\).  Create two seed variables \(v\) and \(h\), in
vertical-first order, each with weight \(1/2\).  Their common alphabet is
\(\Sigma=2^{[N]}\), and \(c(A)=|A|\).
A label specifies all candidates selected from the corresponding candidate
set.

Use each square \(s\in S\) as the name of a binary constraint.
If \(I_s,J_s\) are its vertical and horizontal traces, place on the common
ordered scope \((v,h)\) the rule
\begin{equation}
 \cR_s=\{(A,B)\in\Sigma^2:
 A\cap I_s\ne\varnothing\ \text{or}\ B\cap J_s\ne\varnothing\}.
\label{eq:square-relation}
\end{equation}
The seed language has \(\cL_2=\{\cR_s:s\in S\}\), with all other arity
components empty; the constraint set has one constraint per square.  Feasible assignments are exactly
candidate-line covers, so the integral CSP optimum is \(\kappa=K/2\).
Because both traces are nonempty, \(\Top=[N]\) is absorbing.

\subsection{A connected LP solution}
\label{sec:connected-square}

The transfer theorem requires more than a fractional cover: it needs a
rational basic-LP solution with fixed variable marginals and consistent,
connected local distributions.  We construct the variable marginals
\(\mu_v,\mu_h\) first, and then define, for every square \(s\), an explicit
coupling \(\lambda_s\) having exactly those marginals.

Write \(x_j^u\in[0,1]\) for the fractional weight of candidate \(j\) in
orientation \(u\in\{v,h\}\), and set
\(x^u(T)=\sum_{j\in T}x_j^u\).  Let
\[
 F_0^u=0,\qquad F_j^u=\sum_{k=1}^j x_k^u.
\]
For \(\theta\in[0,1)\), define the sampled label
\begin{equation}
 A_u(\theta)=
 \{j\in[N]:(F_{j-1}^u,F_j^u]\cap(\theta+\mathbb Z)
                                      \ne\varnothing\}.
\label{eq:ordered-sampler}
\end{equation}
Write \(\operatorname{Leb}\) for Lebesgue measure on \([0,1)\), equivalently
the uniform probability measure on this interval.  The marginal distribution
of variable \(u\) is
\begin{equation}
 \mu_u^0(A)=
 \operatorname{Leb}\{\theta\in[0,1):A_u(\theta)=A\}
 \qquad(A\subseteq[N]).
\label{eq:square-marginal}
\end{equation}
Equivalently, sample a uniform \(\theta\) and output \(A_u(\theta)\).
For every \(j\),
\[
 \sum_{A\ni j}\mu_u^0(A)=x_j^u,
 \qquad
 \sum_A\mu_u^0(A)|A|=\sum_jx_j^u.
\]
If \(T=\{a,\ldots,a+w-1\}\), its mass intervals concatenate into
\((F_{a-1}^u,F_{a+w-1}^u]\).  Hence
\[
 \Prb_{A\sim\mu_u^0}[A\cap T\ne\varnothing]
 =\min\{1,x^u(T)\}.
\]

We now construct a local distribution consistent with these two
marginals.  Fix a square \(s\), write
\(I_s=\{a_s,\ldots\}\), \(J_s=\{b_s,\ldots\}\), and set
\[
 \alpha_s=x^v(I_s),\qquad \beta_s=x^h(J_s).
\]
Feasibility gives \(\alpha_s+\beta_s\ge1\).  For uniform
\(U\in[0,1)\), define
\begin{equation}
 \theta_v^s(U)=(U+F^v_{a_s-1})\bmod1,\qquad
 \theta_h^s(U)=(U+F^h_{b_s-1}-\alpha_s)\bmod1,
\label{eq:square-rotation}
\end{equation}
where each residue is taken in \([0,1)\), and put
\[
 A_s(U)=A_v(\theta_v^s(U)),\qquad
 B_s(U)=A_h(\theta_h^s(U)).
\]
Define the base local distribution explicitly by
\begin{equation}
 \lambda_s^0(A,B)=
 \operatorname{Leb}\{U\in[0,1):
                A_s(U)=A,\ B_s(U)=B\}.
\label{eq:square-local}
\end{equation}
Each shifted parameter is uniform modulo one, so the first marginal of
\(\lambda_s^0\) is \(\mu_v^0\) and the second is \(\mu_h^0\).  Moreover,
\(A_s(U)\) hits \(I_s\) when \(U+\mathbb Z\) meets
\((0,\alpha_s]\), while \(B_s(U)\) hits \(J_s\) when it meets
\((\alpha_s,\alpha_s+\beta_s]\).  These consecutive intervals have total
length at least one.  Thus every pair in the support of \(\lambda_s^0\)
satisfies \(\cR_s\).  All breakpoints are rational, so
\(\mu_u^0\) and \(\lambda_s^0\) are rational distributions.

\begin{figure}[t]
\centering
\begin{tikzpicture}[x=10cm,y=.62cm,font=\scriptsize]
  \draw[-{Latex}] (0,0)--(1.04,0);
  \draw (0,-.08)--(0,.08) node[below=3pt] {\(0\)};
  \draw (1,-.08)--(1,.08) node[below=3pt] {\(1\equiv0\)};
  \filldraw[fill=vblue!25,draw=vblue]
    (0,1.15) rectangle (.58,1.65);
  \node at (.29,1.4) {vertical trace hit};
  \node[vblue,above=2pt] at (.29,1.65) {\((0,\alpha_s]\)};
  \filldraw[fill=horange!25,draw=horange]
    (.58,.35) rectangle (1,.85);
  \filldraw[fill=horange!25,draw=horange]
    (0,.35) rectangle (.18,.85);
  \node[anchor=west] at (.60,.6) {horizontal trace hit};
  \node[anchor=west] at (.015,.6) {wrap};
  \draw[densely dashed] (.58,.16)--(.58,1.78)
    node[above] {\(\alpha_s\)};
  \draw[densely dashed] (.18,.16)--(.18,1.02);
  \node[below=3pt] at (.18,0)
    {\(\alpha_s+\beta_s-1\)};
  \fill (.78,0) circle (1.5pt)
    node[below=3pt] {sample \(U\)};
\end{tikzpicture}
\caption{The coupling on \([0,1)\), in the representative case
\(0<\alpha_s,\beta_s<1\) and \(\alpha_s+\beta_s>1\).
The orange arc is \((\alpha_s,\alpha_s+\beta_s]\) modulo one.  Since the
two arcs cover the circle, every value of \(U\) hits at least one trace.}
\label{fig:sampler-coupling}
\end{figure}

It remains to connect the supports without changing consistency.  Choose a
small rational \(\rho\in(0,1)\) and define
\[
 q_\rho(C\mid A)
 =(1-\rho)\one[C=A]+\rho\one[C=\Top].
\]
The perturbed variable marginals and local distributions are
\begin{align*}
 \mu_u^\rho(C)
   &=\sum_A\mu_u^0(A)q_\rho(C\mid A)
     =(1-\rho)\mu_u^0(C)+\rho\delta_{\Top}(C),\\
 \lambda_s^\rho(C,D)
   &=\sum_{A,B}\lambda_s^0(A,B)
       q_\rho(C\mid A)q_\rho(D\mid B).
\end{align*}
The two marginals of \(\lambda_s^\rho\) are therefore exactly
\(\mu_v^\rho\) and \(\mu_h^\rho\).  Absorption keeps every sampled pair
feasible, and every base pair has the supported path
\[
          (A,B)\longrightarrow(\Top,B)\longrightarrow(\Top,\Top).
\]
Every perturbed pair likewise has a supported path to
\((\Top,\Top)\), obtained by replacing its remaining coordinates by
\(\Top\).  Thus every \(\lambda_s^\rho\) has connected support.  The resulting LP
value is \(\ell_\rho=(1-\rho)L'/2+\rho N\).
For sufficiently small \(\rho\), \(\kappa/\ell_\rho\) remains arbitrarily
close to \(K/L'\).

\subsection{The hard CSP}

Fix \(\eps\), the seed, and the perturbation and transfer parameters
before choosing the Unique Games alphabet.  They remain constant as the
Unique Games input grows.  Apply \cref{thm:transfer} and enumerate output
variables by type as
\[
 V_{\rm out}^{\rm V}=\{v_1,\ldots,v_p\},\quad \seed(v_i)=v,
 \qquad
 V_{\rm out}^{\rm H}=\{h_1,\ldots,h_q\},\quad \seed(h_j)=h.
\]
Each output variable \(u\) retains alphabet \(2^{[N]}\) and label cost
\(c(A)=|A|\), so its objective contribution is \(w_u|A|\).
Every output constraint has scope \((v_i,h_j)\).  If its
seed constraint is \(s\), it uses precisely \(\cR_s\), with the inherited
ordered trace pair \((I_s,J_s)\).  Hence the hard CSP is a bipartite
constraint multigraph whose edges are labeled by seed squares.  Several
distinct constraints may have the same scope.

\subsection{Removing rational weights}
\label{sec:cloning}

Unit-cost stabbing requires an unweighted sum of label cardinalities.
The following cloning construction removes the rational variable weights
while preserving rules and their geometric interpretation.

\begin{lemma}[Cartesian cloning]
\label{lem:cloning}
Let \(I\) be a strict-CSP instance with \(n\) variables, nonnegative
rational weights \(w_u\) summing to one, and label costs in \([0,N]\).
For every rational \(\eta>0\), one can construct an unweighted instance
\(\widehat I\) such that, for the
unnormalized sum of clone costs,
\[
 D\opt(I)\le\opt(\widehat I)\le D(\opt(I)+\eta),
 \qquad D=\left\lceil\frac{nN}{\eta}\right\rceil.
\]
Rules, variable types, parent constraints, and coordinate orders are
preserved.  The construction is polynomial-time for fixed maximum arity,
\(N\), and \(\eta\).
\end{lemma}

\begin{proof}
Set \(m_u=\max\{1,\lceil Dw_u\rceil\}\) and replace each variable \(u\)
by clones \(u^{(1)},\ldots,u^{(m_u)}\).  More explicitly, if an original
constraint \(f\) has ordered scope \((u_1,\ldots,u_r)\), then, for every
choice of clone indices
\((a_1,\ldots,a_r)\in[m_{u_1}]\times\cdots\times[m_{u_r}]\), create a
constraint with the same rule as \(f\) on the ordered scope
\[
              (u_1^{(a_1)},\ldots,u_r^{(a_r)}).
\]
This collection is the full Cartesian family associated with \(f\).

Given a feasible assignment of \(I\), assign every clone of \(u\) the
original label of \(u\).  Each new constraint then sees exactly the label
tuple of its parent constraint and is satisfied.  Repeating an optimal
assignment this way costs at most
\(D\opt(I)+nN\le D(\opt(I)+\eta)\).

Conversely, start with any feasible assignment of the cloned instance and,
for every original variable \(u\), choose a clone of \(u\) having minimum
label cost.  For each original constraint \(f\), the cloned instance
contains the member of its Cartesian family formed by precisely these
chosen clones.  Since that member is satisfied, the chosen labels satisfy
\(f\).  They therefore define a feasible assignment of \(I\).  Because
\(m_u\ge Dw_u\), the sum of all clone costs is at least \(D\) times the
cost of this original assignment, which proves the lower bound.
Normalizing clone weights divides all costs by one common factor.

The total number of clones is
\(M=\sum_u m_u\le D+n\).  If the maximum arity is \(r_0\), there are at
most \(|C|M^{r_0}\) cloned constraints.  For fixed \(N,\eta,r_0\), this is
polynomial in the input size.  Computing \(\lceil Dw_u\rceil\) requires
only rational arithmetic of polynomial bit complexity, not a common
denominator for all weights.
\end{proof}

Call the clones \(v_i^{(a)}\) and \(h_j^{(b)}\).  A constraint \(g\) with
scope \((v_i,h_j)\) is replaced by the full family
\(g^{a,b}=(v_i^{(a)},h_j^{(b)})\), and every member retains
\(\seed(g^{a,b})=\seed(g)\).

\subsection{Geometric realization}
\label{sec:square-realization}

We now realize each cloned constraint by a square meeting exactly its
two inherited candidate traces.  Disjoint private coordinate intervals
separate candidates belonging to different variables.

Give every vertical clone \(v_i^{(a)}\) the \(N\) candidate lines
\(x=X_{i,a}+k\), \(k\in[N]\), and every horizontal clone
\(h_j^{(b)}\) the \(N\) candidate lines \(y=Y_{j,b}+k\), \(k\in[N]\).
Here \(X_{i,a}\) and \(Y_{j,b}\) are the base coordinates of the two
candidate blocks.  Choose them as consecutive multiples of \(N+1\),
independently in the two orientations.
The figure shows several cloned constraints simultaneously.  Squares of
the same color descend from the same seed square and consequently have the
same size and relative position inside every block pair in which they occur.

\begin{figure}[t]
\centering
\begin{tikzpicture}[x=.53cm,y=.53cm,font=\small,
                    brace/.style={decorate,
                    decoration={brace,amplitude=4pt}},
                    seedone/.style={fill=seedgreen,fill opacity=.18,
                    draw=seedgreen!70!black,draw opacity=1,very thick},
                    seedtwo/.style={fill=seedpurple,fill opacity=.18,
                    draw=seedpurple!75!black,draw opacity=1,very thick},
                    seedthree/.style={fill=seedrose,fill opacity=.20,
                    draw=seedrose!72!black,draw opacity=1,very thick}]
  \foreach \x in {1,1.5,2,2.5,3,3.5,4,4.5,5}
    \draw[vblue!68] (\x,.5)--(\x,11.5);
  \foreach \x in {7,7.5,8,8.5,9,9.5,10,10.5,11}
    \draw[vblue!50] (\x,.5)--(\x,11.5);
  \foreach \y in {1,1.5,2,2.5,3,3.5,4,4.5,5}
    \draw[horange!70] (.5,\y)--(11.5,\y);
  \foreach \y in {7,7.5,8,8.5,9,9.5,10,10.5,11}
    \draw[horange!52] (.5,\y)--(11.5,\y);

  \filldraw[seedone]   (1.25,1.25) rectangle (3.75,3.75);
  \filldraw[seedtwo]   (2.75,2.25) rectangle (4.25,3.75);
  \filldraw[seedthree] (3.25,1.25) rectangle (4.25,2.25);
  \filldraw[seedone]   (7.25,7.25) rectangle (9.75,9.75);
  \filldraw[seedtwo]   (2.75,8.25) rectangle (4.25,9.75);
  \filldraw[seedthree] (9.25,1.25) rectangle (10.25,2.25);

  \draw[gray!50,densely dashed] (6,.5)--(6,11.5);
  \draw[gray!50,densely dashed] (.5,6)--(11.5,6);
  \draw[brace,vblue] (1,11.8)--(5,11.8)
    node[midway,above=5pt] {candidates of \(v_i^{(a)}\)};
  \draw[brace,vblue!75] (7,11.8)--(11,11.8)
    node[midway,above=5pt] {candidates of \(v_{i'}^{(a')}\)};
  \draw[brace,horange] (.25,1)--(.25,5)
    node[midway,left=7pt,anchor=east] {candidates of \(h_j^{(b)}\)};
  \draw[brace,horange!75] (.25,7)--(.25,11)
    node[midway,left=7pt,anchor=east] {candidates of \(h_{j'}^{(b')}\)};

  \filldraw[seedone] (-.1,-.35) rectangle (.25,0)
    node[right=4pt,opacity=1,text=black] at (.25,-.18) {seed \(s_1\)};
  \filldraw[seedtwo] (3.55,-.35) rectangle (3.9,0)
    node[right=4pt,opacity=1,text=black] at (3.9,-.18) {seed \(s_2\)};
  \filldraw[seedthree] (7.2,-.35) rectangle (7.55,0)
    node[right=4pt,opacity=1,text=black] at (7.55,-.18) {seed \(s_3\)};
\end{tikzpicture}
\caption{Several cloned constraints in four pairs of private candidate
blocks.  The three sizes represent three seed squares.  Copies of seed
\(s_r\) have the same color, size, and relative position in every block
pair where that seed occurs; the lower-left pair also illustrates that
different realized squares may overlap.  Block spacing excludes candidates
assigned to every other clone.}
\label{fig:square-realization}
\end{figure}

Consecutive private coordinate intervals contain \(N\) candidates and
leave one unused integer between them.  Together with the quarter-integer
margins below, this prevents a realized square from meeting candidates in
a neighboring interval.

For a cloned constraint with seed square \(s\), write
\[
 I_s=\{a_s,\ldots,a_s+w_s-1\},\qquad
 J_s=\{b_s,\ldots,b_s+w_s-1\}.
\]
Create a square whose two projections are
\[
 [X_{i,a}+a_s-\tfrac14,
  X_{i,a}+a_s+w_s-\tfrac34],\qquad
 [Y_{j,b}+b_s-\tfrac14,
  Y_{j,b}+b_s+w_s-\tfrac34].
\]
They have equal length \(w_s-1/2\), contain exactly the candidates indexed
by \(I_s,J_s\), and meet no candidates assigned to another clone.
Therefore this square is stabbed
exactly when the cloned constraint is satisfied.  CSP labels and chosen
candidate lines correspond bijectively, and all normalized costs change
by one global factor.  Coincident geometric copies may be removed.
There are polynomially many candidates and squares by \cref{lem:cloning}.
The coordinates have magnitude \(O(M(N+1))\) and denominator at most four,
so their binary encoding length is polynomial.

\begin{proof}[Proof of \cref{thm:square-hardness}]
Choose \(t,N\) so that \(K/L'>2-\eps/4\).  Then choose \(\rho\), the
transfer error \(\delta\), and the cloning tolerance \(\eta\), in this
order, sufficiently small that
\[
             \frac{K/2-\delta}{\ell_\rho+\delta+\eta}>2-\eps.
\]
All chosen parameters are constants depending only on \(\eps\).
The transfer theorem, cloning lemma, and geometric realization therefore
give a polynomial reduction proving finite-candidate hardness.
Every realized projection has the quarter-integer margins used in
\cref{sec:square-gap}.  Discarding useless lines and snapping each useful
line to a nearest integer preserves its incidences and lands at a supplied
candidate.  Thus the unrestricted integral optimum is identical;
aggregating and capping fractional masses gives equality for the LPs too.
\end{proof}

\section{Extensions to Cubes and Horizontal Segments}
\label{sec:cubes}
\label{sec:segments}

The proofs for cubes and horizontal segments are very similar to the
proof for squares.  We give only high-level sketches here; the complete
arguments appear in Appendices~\ref{app:cubes} and~\ref{app:segments}.

\pp{Cubes and separated \(d\)-intervals.}
For fixed \(d\ge2\), the harmonic construction extends to \(d\) axes and
gives a cube-stabbing LP gap arbitrarily close to \(d\).  The fractional
cover uses harmonic weights in every orientation, and the integral lower
bound follows from the same multiplicative-growth argument as for
squares, processing the physical gaps on all \(d\) axes.  The seed has
one variable per orientation, with alphabet \(2^{[N]}\) and cost
\(c(A)=|A|\).  The ordered sampler couples the \(d\) trace-hit arcs
consecutively around the unit circle; their total length is at least
one.  An absorbing-label perturbation connects the local supports.
After the transfer theorem and Cartesian cloning, private coordinate
blocks realize every constraint as a cube because its inherited traces
have equal cardinality.  Snapping gives the unrestricted model, proving
\cref{thm:main-cubes}.  Appendix~\ref{app:cubes} contains the gap
construction and the full hardness proof.

Placing the candidates of orientation \(i\) on a separate track \(L_i\)
and replacing each cube by its \(d\) coordinate projections preserves
incidences and costs.  The resulting separated \(d\)-intervals have
nonempty components of equal length, proving the hardness assertion in
\cref{cor:main-dintervals}.  For the matching LP upper bound, assign each
object to a track carrying fractional mass at least \(1/d\), scale the
weights by \(d\), and use integrality of interval covering on each track.

\pp{Horizontal segments.}
The Kovaleva--Spieksma gap family~\cite{KovalevaSpieksma2006} has integral
optimum \(m\) and fractional value \(L_m\), with
\(m/L_m\to e/(e-1)\).  Each segment meets a single candidate horizontal
row and a nonempty consecutive block of candidate vertical columns.
Use a row variable and a column variable, padding their candidate sets
to a common size.  The ordered sampler and absorbing-label perturbation
give the connected basic-LP solution as for squares.  After transfer and
cloning, private rows and column blocks realize each inherited constraint
as a horizontal segment on its designated row spanning exactly its
consecutive columns.  Incidences and costs are preserved, and snapping
vertical lines to candidates gives the unrestricted model.  This proves
\cref{thm:main-segments}; Appendix~\ref{app:segments} gives the gap
construction and the full proof.

\section{Hardness for Hitting Halfspaces}
\label{sec:halfspaces}

We prove \cref{cor:main-box-points,thm:main-halfspaces} by first realizing an unbounded
rectangle-hitting gap as a box-stabbing instance in \(\R^3\).
Two incidence-preserving transformations then replace the candidate lines
by points and the boxes by halfspaces in \(\R^6\).  All problems
in this section use supplied finite candidate sets and unit costs.
In particular, a halfspace-hitting instance consists of a point set
\(\mathcal P\subset\R^6\) and a family \(\mathcal H\) of halfspaces;
the objective is to minimize \(|X|\) over sets \(X\subseteq\mathcal P\)
such that \(X\cap H\ne\varnothing\) for every \(H\in\mathcal H\).

The main difference from square stabbing is that a seed trace is now a
planar rectangle trace, rather than an interval in an ordered list.  We
therefore use product sampling and a coupling of hit indicators instead
of the ordered sampler.  These local supports are already connected, so
no absorbing-label perturbation is needed.  The third coordinate in the
box realization records the second coordinate of each seed point.  The
reductions here establish hardness only for supplied finite candidates.

\subsection{A rectangle-hitting seed}

We use the primal rectangle range spaces of Pach and
Tardos~\cite[Theorem~4]{PachTardos2013}.  Their lower bound on
\(\eps\)-nets gives the following form of the natural-LP gap.

\begin{lemma}[Rectangle-hitting gap]
\label{lem:rectangle-hitting-gap}
For every constant \(M>1\), there are a finite rational point set
\(P=\{p_1,\ldots,p_N\}\subset(0,1)^2\), a finite family
\(\mathcal F\) of axis-parallel rectangles contained in \([0,1]^2\),
and a rational vector \(x\in(0,1)^N\) such that
\[
 \sum_{j:p_j\in F}x_j\ge1\quad(F\in\mathcal F),\qquad
 L=\sum_{j=1}^Nx_j,\qquad K>ML,
\]
where \(K\) is the minimum number of points of \(P\) hitting every
rectangle of \(\mathcal F\).  Every rectangle has positive side lengths,
and no point of \(P\) lies on a rectangle boundary.
\end{lemma}

\begin{proof}
For arbitrarily small \(\eps>0\), the construction of Pach and Tardos
gives a finite point set of size \(N\) for which hitting all rectangles
containing at least \(\eps N\) points requires
\(\Omega(\eps^{-1}\log\log(1/\eps))\) points.  Retain one rectangle
for each distinct such trace.  Taking \(\eps\) rational and the instance
large enough that \(\eps N>1\), the uniform vector
\(x_j=1/(\eps N)\) is feasible and has value \(L=1/\eps\).
The ratio \(K/L\) is therefore arbitrarily large.

The finite incidence pattern can be represented with rational
coordinates: move each rectangle boundary into an adjacent coordinate
gap without changing its trace, perturb the points to make both
coordinate orders strict,
and replace their coordinates by ranks.  Place rectangle endpoints in
the corresponding gaps between ranks.  A common affine scaling puts
the points in \((0,1)^2\) and the rectangles in \([0,1]^2\), preserving
all traces and giving positive side lengths.
\end{proof}

Fix one such instance and write
\(T_F=\{j\in[N]:p_j\in F\}\).  Create two seed variables \(v,h\),
each of weight \(1/2\), with alphabet \(\Sigma=2^{[N]}\) and cost
\(c(A)=|A|\).  For each \(F\in\mathcal F\), put a separate constraint
on the ordered scope \((v,h)\), with rule
\begin{equation}
 \cR_F=\{(A,B)\in\Sigma^2:(A\cup B)\cap T_F\ne\varnothing\}.
\label{eq:rectangle-seed-rule}
\end{equation}
The language has \(\cL_2=\{\cR_F:F\in\mathcal F\}\), with all other
arity components empty.  The label \(\Top=[N]\) is absorbing because
every trace is nonempty.  The union of two feasible labels is a hitting
set for \(\mathcal F\), and any hitting set can be placed entirely on
one variable.  Hence the integral seed optimum is \(\kappa=K/2\).

\begin{lemma}
\label{lem:rectangle-seed-lp}
The seed has a rational feasible basic-LP solution of value \(L\) whose
local supports are connected.
\end{lemma}

\begin{proof}
Recall that \(x_j=1/(\eps N)\) for every \(j\in[N]\), with
\(0<x_j<1\) and \(L=1/\eps\).
Let \(\mu\) be the distribution on \(2^{[N]}\) obtained by including
each \(j\) independently with probability \(x_j\).  Use \(\mu\) as
the marginal of both seed variables.  It has full support and
\(\E_{A\sim\mu}|A|=L\).

For a fixed rectangle \(F\), set
\[
 h_F=\Prb_{A\sim\mu}[A\cap T_F\ne\varnothing]
     =1-\prod_{j\in T_F}(1-x_j).
\]
Fractional feasibility and \(0<x_j<1\) give
\(1>h_F\ge1-e^{-1}>1/2\).  Let \(\mu_F^1\) and \(\mu_F^0\) be
\(\mu\) conditioned on hitting and missing \(T_F\), respectively.
First sample a pair of hit indicators \((H_1,H_2)\in\{0,1\}^2\),
assigning probabilities \(2h_F-1\), \(1-h_F\), and \(1-h_F\) to
\((1,1)\), \((1,0)\), and \((0,1)\), respectively, and probability
zero to \((0,0)\).  Given this pair, independently sample
\(A\sim\mu_F^{H_1}\) and \(B\sim\mu_F^{H_2}\); this defines
\(\lambda_F\).  Each \(H_i\) equals one with probability \(h_F\), so
each label has marginal
\(h_F\mu_F^1+(1-h_F)\mu_F^0=\mu\).  Also \(H_1+H_2\ge1\), hence
\((A\cup B)\cap T_F\ne\varnothing\) in every sample.

All three indicator probabilities are positive, and \(\mu\) has full
support.  Thus every pair in \(\cR_F\) has positive probability under
\(\lambda_F\).  Its support is connected: every supported pair has
a path to \((\Top,\Top)\), obtained by replacing its coordinates by
\(\Top\) one at a time.  All probabilities are rational.  The basic-LP
objective is \(\tfrac12L+\tfrac12L=L\).
\end{proof}

\subsection{Stabbing 3D boxes with lines}

\begin{proof}[Proof of \cref{thm:box-lines}]
Choose the seed in \cref{lem:rectangle-hitting-gap} with
\(K/(2L)>C\).  By \cref{lem:rectangle-seed-lp}, the hypotheses of
\cref{thm:transfer} hold with \(\kappa=K/2\) and \(\ell=L\).
Choose positive rational transfer and cloning errors \(\delta,\eta\)
such that
\begin{equation}
 0<\delta<(K/2-L)/2,\qquad
 \frac{K/2-\delta}{L+\delta+\eta}>C.
\label{eq:box-gap-choice}
\end{equation}
Apply \cref{thm:transfer,lem:cloning}.  For the cloning scale \(D\),
the resulting unweighted CSP has the gap
\[
 \text{YES: }\opt\le D(L+\delta+\eta),\qquad
 \text{NO: }\opt\ge D(K/2-\delta).
\]
Every output variable has parent \(v\) or \(h\).  Every constraint
on a type-\(v\) variable \(u\) and a type-\(h\) variable \(w\)
retains a seed rectangle \(F\) and exactly the rule
\cref{eq:rectangle-seed-rule}.  We retain different seed predicates
even when they have the same output scope.

Write \(p_j=(a_j,b_j)\) and \(F=I_F\times J_F\).
Assign a distinct multiple \(Y_u\) of three to each type-\(v\)
variable, and a distinct multiple \(X_w\) of three to each type-\(h\)
variable.  Supply the candidate lines
\begin{align*}
 \lambda^x_{u,j}&=\{(t,Y_u+a_j,b_j):t\in\R\},\qquad j\in[N],\\
 \lambda^y_{w,j}&=\{(X_w+a_j,t,b_j):t\in\R\},\qquad j\in[N].
\end{align*}
Represent a constraint \((u,w,F)\) by the box
\begin{equation}
 B_{u,w,F}=(X_w+I_F)\times(Y_u+I_F)\times J_F.
\label{eq:box-realization}
\end{equation}
Its first two side lengths are equal.  The private blocks are disjoint,
so the only candidate \(x\)-lines that can meet this box belong to
\(u\), and the only candidate \(y\)-lines that can meet it belong to
\(w\).  For those candidates,
\[
 \lambda^x_{u,j}\cap B_{u,w,F}\ne\varnothing
 \quad\Longleftrightarrow\quad j\in T_F
 \quad\Longleftrightarrow\quad
 \lambda^y_{w,j}\cap B_{u,w,F}\ne\varnothing.
\]
Consequently, a set of selected lines corresponds exactly to the subset
label at every variable.  Feasibility and the unnormalized sum of label
cardinalities agree with stabbing feasibility and the number of selected
lines.  Distinct seed points give distinct candidates within each block.

For fixed \(C\), the seed and all error parameters are fixed before
the Unique Games instance is chosen.  Cloning and the geometric
construction have polynomial size; the block offsets and all coordinates
have polynomial bit length.  The gap in
\cref{eq:box-gap-choice} proves the theorem.
\end{proof}

\subsection{Hitting 3D boxes with points}

The next transformation applies to any finite box-stabbing instance
whose candidate lines have the two orientations used above.

\begin{lemma}[Two-plane realization]
\label{lem:box-lines-to-points}
A box-stabbing instance in \(\R^3\) with supplied \(x\)-parallel
and \(y\)-parallel candidate lines can be transformed in polynomial time
into a box-hitting instance in \(\R^3\) with supplied candidate points
on two parallel planes.  The transformation gives bijections between
candidates and between demands, preserves every incidence, and hence
preserves the objective value of every candidate subset.
\end{lemma}

\begin{proof}
Choose increasing rational affine maps \(f_x,f_y\) sending all relevant
\(x\)- and \(y\)-coordinates, respectively, into \((0,1/4)\).
Include both box endpoints and the fixed coordinates of candidate
lines among these coordinates.  Such maps have polynomial bit complexity.
Replace the candidates by
\begin{align*}
 \{(t,y,z):t\in\R\}&\longmapsto(f_y(y),-f_y(y),z),\\
 \{(x,t,z):t\in\R\}&\longmapsto(1+f_x(x),1-f_x(x),z).
\end{align*}
The two images lie on the planes \(x+y=0\) and \(x+y=2\).
They are disjoint, and each map is injective.

For a box
\(B=[x_0,x_1]\times[y_0,y_1]\times[z_0,z_1]\), put
\[
 a=f_y(y_0),\quad b=f_y(y_1),\quad
 c=f_x(x_0),\quad d=f_x(x_1)
\]
and replace it by
\begin{equation}
 Q_B=[a,1+d]\times[-b,1-c]\times[z_0,z_1].
\label{eq:box-two-plane-box}
\end{equation}
For a candidate point \((t,-t,z)\) on the first plane, the upper bounds in the
first two coordinates are automatic, and the lower bounds say
\(a\le t\le b\).  Thus \(Q_B\) contains its image precisely when
the original \(x\)-line meets \(B\).
For a candidate point \((1+s,1-s,z)\) on the second plane, the lower bounds are
automatic, and the upper bounds say \(c\le s\le d\).  This is precisely
the condition for the original \(y\)-line to meet \(B\).
Both tests also require \(z_0\le z\le z_1\).
All incidences, cardinalities, and rational encoding bounds are preserved.
\end{proof}

\begin{proof}[Proof of \cref{cor:main-box-points}]
Apply \cref{lem:box-lines-to-points} to the instances of
\cref{thm:box-lines}.  The transformation preserves every incidence and
the cardinality of every candidate subset.
\end{proof}

\subsection{Hitting Halfspaces with points in \texorpdfstring{\(\mathbb{R}^6\)}{R6}}

\begin{lemma}[Six-dimensional lifting]
\label{lem:boxes-to-halfspaces-six}
Every finite feasible box-hitting instance in \(\R^3\) has a
polynomial-time, incidence-preserving realization by candidate points
and halfspaces in \(\R^6\).  Candidate subsets correspond bijectively
and have the same cardinality.
\end{lemma}

\begin{proof}
We use the rank-separation construction underlying
Pach and Tardos~\cite[Lemma~2.3]{PachTardos2013}, after introducing
reciprocal coordinates.
Let \(\mathcal P\subset\R^3\) be the candidate set and
\(n=|\mathcal P|\).  In each coordinate, replace the distinct candidate
values by their ranks, retaining ties.  Write \(r_i(p)\in[n]\) for
the rank of the \(i\)th coordinate of \(p\).
Since every input box \(Q\) contains a candidate, there are rank bounds
\(\ell_i(Q)\le u_i(Q)\) such that
\[
 p\in Q\quad\Longleftrightarrow\quad
 \ell_i(Q)\le r_i(p)\le u_i(Q)\quad(i\in[3]).
\]
Map each candidate to
\begin{equation}
 \varphi(p)=
 (8^{r_1(p)},8^{-r_1(p)},
  8^{r_2(p)},8^{-r_2(p)},
  8^{r_3(p)},8^{-r_3(p)})
\label{eq:halfspace-candidate-lift}
\end{equation}
and replace \(Q\) by the closed halfspace
\begin{equation}
 H_Q=\left\{y\in\R^6:
   \sum_{i=1}^3
     \bigl(8^{-u_i(Q)}y_{2i-1}
                +8^{\ell_i(Q)}y_{2i}\bigr)\le7\right\}.
\label{eq:halfspace-range-lift}
\end{equation}
If \(p\in Q\), each of the six summands at \(\varphi(p)\) is at most
one, so the sum is at most six.  If \(p\notin Q\), one rank inequality
fails by at least one.  The corresponding summand is then at least
eight, and all other summands are positive.  Therefore
\[
                  \varphi(p)\in H_Q
                  \quad\Longleftrightarrow\quad p\in Q.
\]
The map \(\varphi\) is injective.  Every exponent has magnitude at most
\(n\), so the rational coordinates and halfspace coefficients have
\(O(n)\) bits.  Computing the ranks and writing the transformed instance
takes polynomial time.
\end{proof}

\begin{proof}[Proof of \cref{thm:main-halfspaces}]
Fix \(C>1\). Start with the NP-hard candidate-line box instances
from \cref{thm:box-lines}.  Apply
\cref{lem:box-lines-to-points} and then
\cref{lem:boxes-to-halfspaces-six}.  Each selected line becomes one
supplied point, and each stabbing constraint becomes one halfspace-hitting
constraint with exactly the same incident candidates.  Both transformations
are polynomial-time and preserve the optimum.  A \(C\)-approximation
for the resulting halfspace-hitting instances would consequently give a
\(C\)-approximation for the box instances, proving the claim.
\end{proof}

The candidate restriction is essential.  All halfspaces in
\cref{eq:halfspace-range-lift} contain the origin, whereas every supplied
point in \cref{eq:halfspace-candidate-lift} has positive coordinates.
Allowing arbitrary hitting points would make these particular instances
solvable with the single point at the origin.

\section{Conclusion}
\label{sec:conclusion}

We proved tight NP-hardness bounds matching the orientation-rounding
factors for arbitrary-size squares and fixed-dimensional cubes,
and matching \(e/(e-1)\) for horizontal segments.
The reduction uses ordered consecutive traces,
a feasible coupling with prescribed marginals, an absorbing label, and a
geometric realization preserved by cloning.  It also gives tight
hardness for separated \(d\)-intervals when every component is nonempty
and the \(d\) components of each object have equal length.

The same transfer theorem proves that stabbing three-dimensional boxes
with supplied candidate lines is NP-hard to approximate within any
fixed constant factor. Incidence-preserving algebraic maps transfer
this hardness to hitting boxes in \(\R^3\) and halfspaces in \(\R^6\)
with candidate points. These results concern the finite-candidate model.

For fixed \(d,\eps\), the cube construction fixes \(N\) and uses cube side
lengths \(w-1/2\), \(w\in[N]\).  It therefore already yields near-tight
hardness for a bounded number of side lengths and a side-length ratio
at most \(2N-1\), where the bound depends on \(d,\eps\).
A remaining direction is to determine the approximation threshold as a
function of a prescribed side-length ratio or a prescribed number of
distinct lengths, starting from the improved bound for congruent squares.

\paragraph*{Use of generative AI.}
The authors used ChatGPT iteratively in developing the results presented
in this paper. Initially, they used it to extend an integrality-gap
construction for rectangle stabbing to arbitrary-size square stabbing,
obtaining a family of instances whose integrality gap tends to \(2\).
They then used ChatGPT to apply the KMTV framework to this family and
derive a hardness bound under the Unique Games Conjecture. The authors
substantially simplified and revised the resulting arguments and adapted
the underlying ideas to other geometric stabbing problems. The authors have rewritten and verified the proofs and
take full responsibility for their correctness.


\appendix

\section{Proof of the Transfer Theorem}
\label{app:kmtv}

We prove \cref{thm:transfer} using the KMTV composition
\cite[Sections~6--7]{KMTV2011}, its cost-minimization form
\cite[Theorem~5.4.1]{Manokaran2012}, and the strong-left Unique Games
source of OpenAI~\cite[Lemma~A.1]{OpenAI2026UGT}.
The imported analytic input is connected-space positivity, stated as
\cref{lem:positivity}. The KMTV numbering refers to the linked
author-hosted full version.

Divide all label costs and the desired additive error by
\(\max\{1,\max_{a\in\Sigma}c(a)\}\), reusing the symbols
\(\ell,\kappa,\delta\) for the scaled values.  We may therefore assume
\(0\le c(a)\le1\).  Let \(m=|V|\), \(q=|\Sigma|\), and let \(w_v\),
\(\mu_v\), and \(\lambda_f\) denote the seed weights, marginals, and
local distributions.  Fix \(0<\delta<(\kappa-\ell)/2\).
The parameters are chosen in the order
\(\delta\), \(\gamma=\delta/(4q)\), \(D,\tau\), \(\eta,\zeta\), and \(R\).
Here \cref{lem:positivity} supplies common \(D\ge1\) and \(0<\tau\le1\)
for the finitely many seed distributions, independently of \(R\).  Choose
\[
 0<\eta\le\delta,\qquad
 0<\zeta<\frac{\delta\tau^2}{8mD}.
\]
Only then choose the Unique Games alphabet size \(R\).
All these parameters are independent of the source instance size.

\subsection{Unique Games source}

For the constants \(\eta,\zeta\) chosen above, apply
OpenAI~\cite[Lemma~A.1]{OpenAI2026UGT} with its parameters
\((\zeta,\nu,d,H)=(\eta,\zeta,1,1)\).  It supplies an NP-hard family
of simple bipartite Unique Games instances with left vertices \(\cU\),
right vertices \(\cB\), a fixed alphabet \([R]\), and equal left degrees.
Give every edge weight \(\omega_e=1/|E|\), and set
\(p_a=\sum_{e\ni a}\omega_e\); thus \(p_u=1/|\cU|\).
An edge \(e=(u,b)\) carries a permutation \(\pi_e:[R]\to[R]\);
a labeling \(\Lambda\) satisfies it when
\(\Lambda(b)=\pi_e(\Lambda(u))\).

In the YES case there are a labeling and a set
\(\cU_0\subseteq\cU\) of weight at least \(1-\eta\) such that every edge
incident to \(\cU_0\) is satisfied.  In the NO case, every labeling
satisfies edge weight at most \(\zeta\).  Delete zero-weight edges and
isolated vertices.

\subsection{Clouds and constraints}

For every \(u\in\cU\), \(v\in V\), and \(z\in\Sigma^R\), create an output
variable \((u,v,z)\), set \(\seed(u,v,z)=v\), and give it weight
\begin{equation}
 W(u,v,z)=p_uw_v\prod_{j=1}^R\mu_v(z_j).
\label{eq:cloud-weight}
\end{equation}
The label set and cost remain \(\Sigma,c\).  These rational weights sum
to one.  For an edge \(e=(u,b)\) and a word
\(z=(z_1,\ldots,z_R)\), let \(T_ez\) be the word obtained by reordering
coordinates according to the edge permutation:
\((T_ez)_s=z_{\pi_e(s)}\).  Consequently, if \(e\) is satisfied by
\(\Lambda\), then
\((T_ez)_{\Lambda(u)}=z_{\pi_e(\Lambda(u))}=z_{\Lambda(b)}\).

Fix a right vertex \(b\), a seed constraint \(f\) with scope
\((v_1,\ldots,v_r)\), an ordered tuple of incident edges
\(e_i=(u_i,b)\) (repetitions are allowed), and words
\(z^1,\ldots,z^r\in\Sigma^R\) satisfying
\begin{equation}
 (z^1_j,\ldots,z^r_j)\in\supp(\lambda_f)
                         \quad(j\in[R]).
\label{eq:supported-columns}
\end{equation}
For every such record, create an output constraint with rule \(\cR_f\)
and ordered scope
\begin{equation}
 ((u_1,v_1,T_{e_1}z^1),\ldots,
   (u_r,v_r,T_{e_r}z^r)).
\label{eq:output-constraint}
\end{equation}
Call this output constraint \(g\) and set \(\seed(g)=f\).  Hence
the \(i\)th output variable has parent \(v_i\), and the rule and
coordinate order are copied exactly.  The full record identifies \(g\) and its parent \(f\).
Since the variables in a seed scope are
distinct, so are the output variables in \cref{eq:output-constraint}.
The seed, \(q\), and \(R\) are constants, so the construction is
polynomial.  This also proves all structural assertions in
\cref{thm:transfer}.

\subsection{Completeness}

For \(u\in\cU_0\), label \((u,v,z)\) by \(z_{\Lambda(u)}\); outside
\(\cU_0\), use \(\Top\).  A constraint touching a bad cloud is satisfied
by absorption.  Otherwise its \(i\)th label is
\[
 (T_{e_i}z^i)_{\Lambda(u_i)}
 =z^i_{\pi_{e_i}(\Lambda(u_i))}
 =z^i_{\Lambda(b)}.
\]
These labels form column \(\Lambda(b)\) of
\cref{eq:supported-columns}, so the constraint is satisfied.  Averaging a
dictator under \(\mu_v^{\otimes R}\) gives the seed LP cost.  Therefore
the total cost is at most
\[
 \ell\sum_{u\in\cU_0}p_u+c(\Top)\sum_{u\notin\cU_0}p_u
 \le\ell+\eta.
\]
Taking \(\eta\le\delta\) proves the completeness threshold.

\subsection{Low-degree influence}

On a finite product space, write the orthogonal decomposition
\(f=\sum_{S\subseteq[R]}f_S\), where
\[
 f_S=\sum_{T\subseteq S}(-1)^{|S\setminus T|}\E[f\mid z_T].
\]
The degree-at-most-\(D\) influence of coordinate \(j\) is
\begin{equation}
 \Inf_j^{\le D}(f)=
 \sum_{\substack{S\ni j\\|S|\le D}}\|f_S\|_2^2.
\label{eq:influence}
\end{equation}
Orthogonality gives
\(\sum_j\Inf_j^{\le D}(f)\le D\E[f^2]\).

\begin{lemma}[Connected-space positivity]
\label{lem:positivity}
Fix a finite distribution \(\lambda\) on \(\Sigma^r\) with
connected support.  For every \(\gamma>0\), there are
\(D\ge1\) and \(\tau>0\), independent of \(R\), such that the following
holds.  If \(g_i:\Sigma^R\to[0,1]\) has mean at least \(\gamma\) under
the \(i\)th product marginal of \(\lambda^{\otimes R}\), and
\(\Inf_j^{\le D}(g_i)\le\tau\) for all \(i,j\), then
\[
 \E_{(z^1,\ldots,z^r)\sim\lambda^{\otimes R}}
                   \prod_{i=1}^rg_i(z^i)>0.
\]
\end{lemma}

This follows from Mossel's Gaussian positivity bound
\cite[Proposition~6.4]{Mossel2010}, using connectedness to bound the
correlation away from one \cite[Lemma~2.9]{Mossel2010}; see also
KMTV~\cite[Corollary~7.15]{KMTV2011}.  Restrict each marginal alphabet to
its positive support.  A singleton marginal makes the corresponding
function constant; its value is at least \(\gamma\), so this coordinate
can be removed.  The remaining finite distribution has a positive
minimum atom probability and retains connected support.  The cited
proposition therefore gives \(D,\tau\) independent of \(R\).
If all marginals are singletons, the conclusion is immediate.
For the finitely many seed distributions, take the maximum of the
required degrees and the minimum of the influence thresholds, decreasing
\(\tau\) to at most one if necessary.

\subsection{Soundness}

Let \(F(u,v,z)\in\Sigma\) be a feasible output assignment and set
\(f_{u,v,a}(z)=\one[F(u,v,z)=a]\).
For a right vertex \(b\), average neighboring permuted indicators:
\begin{equation}
 g_{b,v,a}(z)=
 \E_{e=(u,b)\mid b}f_{u,v,a}(T_ez),
\label{eq:neighbor-average}
\end{equation}
where edge \(e\) has conditional probability \(\omega_e/p_b\).  Define
\begin{equation}
 \beta_b=\sum_{v\in V}w_v\sum_{a\in\Sigma}c(a)
           \E_{z\sim\mu_v^{\otimes R}}g_{b,v,a}(z).
\label{eq:local-beta}
\end{equation}
Coordinate permutations preserve product marginals; reversing the edge
and cloud averages shows that the assignment cost is
\(\beta=\sum_bp_b\beta_b\).

Use the constants \(\gamma,D,\tau\) fixed above.  Suppose first that,
at \(b\), all functions
\(g_{b,v,a}\) have every degree-\(D\) influence at most \(\tau\).  Define
\[
 H_v=\{a\in\Sigma:
       \E_{\mu_v^{\otimes R}}g_{b,v,a}\ge\gamma\}.
\]
Every \(H_v\) is nonempty.  For any seed constraint \(f\) with scope
\((v_1,\ldots,v_r)\) and labels \(a_i\in H_{v_i}\), positivity applied
to \(g_{b,v_i,a_i}\) gives a supported word tuple and incident edges on
which \(F\) outputs exactly \((a_1,\ldots,a_r)\).  The corresponding
constraint \cref{eq:output-constraint} is present and satisfied; hence
\[
                  H_{v_1}\times\cdots\times H_{v_r}
                  \subseteq\cR_f.
\]
Choosing the cheapest label in every \(H_v\) therefore gives a feasible
seed assignment.  Put \(s_v=\min_{a\in H_v}c(a)\).
Labels outside \(H_v\) have total probability at most \(q\gamma\), so
\[
 \sum_a c(a)\E g_{b,v,a}
 \ge s_v(1-q\gamma)\ge s_v-q\gamma,
\]
using \(0\le s_v\le1\).  Multiply by \(w_v\), sum over \(v\), and use
\(\sum_vw_v=1\) and \(q\gamma=\delta/4\) to obtain
\begin{equation}
                         \kappa\le\beta_b+\delta/4.
\label{eq:seed-from-low-influence}
\end{equation}

Assume now, for contradiction, that \(\beta<\kappa-\delta\).  Since
\(\kappa\le1\), the right vertices with
\(\beta_b<\kappa-\delta/2\) have total weight greater than
\(\delta/2\); otherwise
\[
 \beta\ge(1-\delta/2)(\kappa-\delta/2)\ge\kappa-\delta.
\]
At every such vertex, \cref{eq:seed-from-low-influence} fails.  Thus some
\(g_{b,v,a}\) has \(\Inf_{j_b}^{\le D}(g_{b,v,a})>\tau\).  Label \(b\)
by one such \(j_b\).

For a left vertex \(u\), define
\[
 S_u=\bigcup_{v,a}
 \{s:\Inf_s^{\le D}(f_{u,v,a})>\tau/2\}.
\]
For each \(u,v\), the functions \(f_{u,v,a}\) over labels \(a\) are
disjoint indicators summing to one.  Thus
\[
 \sum_{v,a}\sum_s\Inf_s^{\le D}(f_{u,v,a})
 \le D\sum_{v,a}\E[f_{u,v,a}^2]=mD,
\]
which gives
\begin{equation}
                              |S_u|\le2mD/\tau.
\label{eq:Su-size}
\end{equation}
Jensen's inequality for the orthogonal components yields
\[
 \tau<\Inf_{j_b}^{\le D}(g_{b,v,a})
 \le\E_{e=(u,b)\mid b}
       \Inf_{\pi_e^{-1}(j_b)}^{\le D}(f_{u,v,a}).
\]
Because every influence is at most one, conditional edge weight at least
\(\tau/2\) satisfies \(\pi_e^{-1}(j_b)\in S_u\).

Label each left vertex uniformly from its nonempty \(S_u\), and label all
remaining vertices arbitrarily.  By \cref{eq:Su-size}, the expected
satisfied edge weight is at least
\[
       \frac\delta2\cdot\frac\tau2\cdot\frac{\tau}{2mD}
       =\frac{\delta\tau^2}{8mD}>0.
\]
The chosen NO-case parameter \(\zeta\) is below this constant, contradicting
the NO-case soundness promise.  Thus every feasible output assignment costs at least
\(\kappa-\delta\).  Together with completeness and the initial rescaling,
this proves \cref{thm:transfer}.

\section{Proofs for Cubes}
\label{app:cubes}

Fix integers \(d\ge2\)
and \(t,N\ge1\).  For every type \(i\in[d]\), place candidates
\(\pi_{i,j}\) at coordinates \(j\in[N]\).  For each integer \(w\in[N]\)
and \(\mathbf a=(a_1,\ldots,a_d)\in[N-w+1]^d\), define the cube
\[
 Q(\mathbf a,w)=\prod_{i=1}^d
 [a_i-\tfrac14,a_i+w-\tfrac34].
\]
Retain it precisely when
\begin{equation}
                     \sum_{i=1}^d a_i-(d-1)\le tw.
\label{eq:cube-retain}
\end{equation}
Its type-\(i\) trace is
\(I_{s,i}=\{a_i,\ldots,a_i+w-1\}\).  Thus every cube has a nonempty
trace of the same cardinality \(w\) in every orientation.  Denote the
integral optimum of the retained family by \(K_{d,t,N}\).

\begin{theorem}[Cube gap]
\label{thm:cube-gap}
The retained cube family has a fractional cover of cost at most
\[
                  L_{d,t,N}=\frac{t+d-1}{d}H_N,
\]
whereas its integral optimum satisfies
\[
                  K_{d,t,N}\ge
                  t\ln\left(1+\frac{N+1}{t+1}\right)-1.
\]
Consequently its asymptotic LP gap is at least
\(dt/(t+d-1)\); the supremum over \(t,N\) is \(d\).  This also holds in
the unrestricted model.
\end{theorem}

\begin{proof}[Proof of \cref{thm:cube-gap}]
Give \(\pi_{i,j}\) weight
\(x_j=\min\{1,(t+d-1)/(d^2j)\}\).  For
\(s=Q(\mathbf a,w)\), before capping the weights at \(1\),
Cauchy--Schwarz yields
\begin{align*}
 \sum_{i=1}^d\sum_{j=a_i}^{a_i+w-1}\frac{t+d-1}{d^2j}
 &\ge \frac{t+d-1}{d^2}
       \frac{(dw)^2}{w(\sum_i a_i+d(w-1)/2)}\\
 &=\frac{(t+d-1)w}{\sum_i a_i+d(w-1)/2}\ge1.
\end{align*}
Indeed, \cref{eq:cube-retain}, \(d\ge2\), and \(w\ge1\) give
\[
 \sum_i a_i+\frac{d(w-1)}2
 \le tw+d-1+\frac{d(w-1)}2
 \le(t+d-1)w.
\]
As before, capping preserves feasibility.  Summing over the \(d\)
orientations gives cost at most \(L_{d,t,N}\).

For the integral bound, fix a cover of size \(K\).  Add sentinels \(0\)
and \(N+1\) on each axis and form physical gaps.  There are \(K+d\) gaps.
Process all gaps, across all axes, one at a time in nondecreasing order,
breaking ties arbitrarily, and stop when one axis is exhausted.  At least
one gap on every other axis then remains.  Thus the number \(q\) processed
is at most \(K+1\), and the sum of their lengths is at least \(N+1\).

Immediately before processing a gap of length \(h\), let \(T\) be the sum
of the lengths of the gaps processed so far.  If \(h>1\), take on each axis
the physical gap of length at least \(h\) that is closest to the origin,
including any already processed gap tied at length \(h\), and let \(a_i\)
be the smallest integer coordinate strictly inside it.
This candidate coordinate is not selected by the cover.  All spatially
preceding gaps are shorter than \(h\), hence already processed, and
\(\sum_i(a_i-1)\le T\).  For every axis, take the \(h-1\) consecutive
unselected coordinates beginning at \(a_i\).  Together they define a
missed cube \(Q(\mathbf a,h-1)\), which cannot belong to the retained
family.  Therefore
\[
 t(h-1)<\sum_i a_i-(d-1)=\sum_i(a_i-1)+1\le T+1.
\]
The same inequality is immediate for \(h=1\).  Thus \(th<T+t+1\), and
the multiplicative-growth argument from \cref{thm:square-gap} gives
\(K\ge t\ln(1+(N+1)/(t+1))-1\).  Letting \(N\) and then \(t\) tend to
infinity yields \(d\); \cref{prop:orientation-rounding} gives the upper
bound.  The quarter-integer margins ensure that every nearest integer
to a point in a projection belongs to its candidate trace; either choice
at a tie is valid.  The snapping and capping argument from
\cref{thm:square-gap} applies coordinatewise and proves the unrestricted
claim.
\end{proof}

\begin{theorem}
\label{thm:cube-hardness}
For every fixed \(d\ge2\) and every \(\eps>0\), unit-cost
stabbing of axis-parallel \(d\)-cubes by coordinate hyperplanes is NP-hard
to approximate within \(d-\eps\), in both models.
\end{theorem}

For a cube \(s\), write its covering rule as
\begin{equation}
 \cR_s=\{(A_1,\ldots,A_d):
       A_i\cap I_{s,i}\ne\varnothing\text{ for some }i\in[d]\}.
\label{eq:cube-relation}
\end{equation}

\begin{proof}[Proof of \cref{thm:cube-hardness}]
Choose a cube-gap instance whose integral and fractional costs \(K,L'\)
satisfy \(K/L'>d-\eps/4\).  The seed has variables
\(u_1,\ldots,u_d\), each of weight \(1/d\), alphabet \(2^{[N]}\), and
cost \(c(A)=|A|\).  A cube \(s\) creates a \(d\)-ary constraint on this
ordered scope with rule \cref{eq:cube-relation}.  The integral seed optimum
is \(K/d\), and \(\Top=[N]\) is absorbing because every trace is nonempty.

Apply the ordered sampler \cref{eq:ordered-sampler} on every axis.  For a
fixed cube let \(\alpha_i=x(I_{s,i})\), set \(b_0=0\) and
\(b_i=\sum_{k=1}^i\alpha_k\), and let \(a_i\) be the first trace index.
Using one uniform \(U\), set
\[
                 \theta_i=U+F^{(i)}_{a_i-1}-b_{i-1}\pmod1.
\]
The \(i\)th trace is hit exactly when \(U+\mathbb Z\) meets
\((b_{i-1},b_i]\).  These intervals are consecutive, and their total
length is \(\sum_i\alpha_i\ge1\); hence every sample satisfies
\cref{eq:cube-relation}.  The distributions are rational.  Independently
replacing each coordinate by \(\Top\) with probability \(\rho\) connects
every tuple in the perturbed support to \((\Top,\ldots,\Top)\), one
coordinate at a time, while preserving all prescribed marginals.
The perturbed LP value is \(\ell_\rho=(1-\rho)L'/d+\rho N\).

The transfer theorem produces \(d\) output-variable types, and every
constraint has one variable of each type in seed order.  Apply Cartesian
cloning.  Give every clone of type \(i\) its own \(N\) type-\(i\)
hyperplanes in a private coordinate interval, with bases at consecutive
multiples of \(N+1\) as in \cref{sec:square-realization}.
If a cloned constraint
inherits cube \(s\), its \(d\) traces have the same cardinality \(w_s\);
take the Cartesian product of the corresponding intervals of length
\(w_s-1/2\).  This is a cube, and private spacing makes it meet exactly
the inherited candidate traces.  Thus the unweighted geometric instance
realizes the hard CSP exactly.

Choose \(\rho,\delta,\eta\) sufficiently small that
\[
           \frac{K/d-\delta}{\ell_\rho+\delta+\eta}>d-\eps.
\]
For fixed \(d,\eps\), the seed, tolerances, and Unique Games alphabet
are constants.  The size bound in \cref{lem:cloning}, with arity \(d\),
and the coordinate bound from \cref{sec:square-realization} show that the
reduction is polynomial.  This proves finite-candidate hardness.
Discarding useless hyperplanes and snapping the others to nearest-integer
candidates gives the unrestricted result.
\end{proof}

\section{Horizontal-Segment Details}
\label{app:segments}

\subsection{The Kovaleva--Spieksma gap construction}

\begin{theorem}[Horizontal-segment gap]
\label{thm:segment-gap}
For every \(m\ge3\), there is a finite-candidate instance in which every
segment meets exactly one horizontal candidate and at least one vertical
candidate, the integral optimum is \(m\), and a fractional cover has value
\(L_m\), where \(m/L_m\to e/(e-1)\).  Consequently the supremum of the
natural-LP integrality gap for interval stabbing is \(e/(e-1)\), in both
the finite-candidate and unrestricted models.
\end{theorem}

\begin{proof}[Proof of \cref{thm:segment-gap}]
Kovaleva and Spieksma introduced the following construction
\cite{KovalevaSpieksma2006}.  Fix \(m\ge3\) and put \(T=m!\).  Take
horizontal candidates \(h_j:y=j\) for \(j\in[m]\) and vertical candidates
\(v_c:x=c\) for \(c\in[T]\).  For every \(j\in[m]\) and \(i\in[j]\),
create the segment
\begin{equation}
 S_{j,i}=
 \left[(i-1)\frac{T}{j}+1,\ i\frac{T}{j}\right]\times\{j\}.
\label{eq:segment-family}
\end{equation}
Thus row \(j\) contains \(j\) pairwise disjoint segments.  Each segment
meets exactly its supporting horizontal candidate and \(T/j\ge1\)
consecutive vertical candidates.

Let \(P=P(m)\) be the least integer in \(\{1,\ldots,m-1\}\) such that
\begin{equation}
 \sum_{r=P+1}^{m}\frac1r\le1
 \le\sum_{r=P}^{m}\frac1r.
\label{eq:segment-P}
\end{equation}
Give every vertical candidate weight \(P/T\), and give row \(j\) weight
\begin{equation}
 z_j=\begin{cases}
       0,&j\le P,\\
       1-P/j,&j>P.
     \end{cases}
\label{eq:segment-fractional}
\end{equation}
Every segment on row \(j\) receives vertical mass \(P/j\), so these
weights form a feasible fractional cover of value
\begin{equation}
 L_m=P+\sum_{j=P+1}^{m}\left(1-\frac Pj\right)
    =m-P\sum_{r=P+1}^{m}\frac1r.
\label{eq:segment-Lm}
\end{equation}

Selecting all \(m\) rows is feasible.  Otherwise, let \(k\) be the largest
unselected row.  All \(m-k\) higher rows are selected, while the \(k\)
disjoint segments on row \(k\) require \(k\) distinct vertical columns.
Thus every integral cover has size at least \(m\).

Put \(H_{P,m}=\sum_{r=P+1}^{m}1/r\).  Equation~\eqref{eq:segment-P}
gives \(1-1/P\le H_{P,m}\le1\).  If \(P\) stayed bounded along an
unbounded sequence of \(m\), this harmonic tail would diverge, a
contradiction.  Thus \(P\to\infty\) and \(H_{P,m}\to1\).
Integral comparison gives
\[
 H_{P,m}\le\log(m/P)\le H_{P,m}+1/P.
\]
Therefore \(P/m\to e^{-1}\), and \cref{eq:segment-Lm} yields
\(L_m/m\to1-1/e\).  Kovaleva and Spieksma proved the matching
\(e/(e-1)\) LP-relative upper bound~\cite{KovalevaSpieksma2006}.

For the unrestricted model, a useful horizontal line is already a
supporting row.  The x-projections of all segments stabbed by a useful
vertical line share its x-coordinate.  Their endpoints are integers,
so their largest left endpoint is an admissible column in every one of
these projections.  Snapping integral lines preserves incidences and
does not increase cost.  For fractional covers, combine masses at the
snapped columns and cap each total at one; feasibility is preserved and
cost does not increase.  The reverse inequalities follow because the
supplied candidates are admissible in the unrestricted model.
\end{proof}

\subsection{Hardness proof}

\begin{theorem}
\label{thm:segment-hardness}
For every \(\eps>0\), unit-cost stabbing of horizontal
segments by horizontal and vertical lines is NP-hard to approximate within
\(e/(e-1)-\eps\), in both the finite-candidate and unrestricted models.
\end{theorem}

\begin{proof}[Proof of \cref{thm:segment-hardness}]
Fix \(0<\eps<1\); larger values follow from any smaller positive error.
Choose \(m\) so that the instance in \cref{thm:segment-gap} has integral
optimum \(K=m\), rational fractional value \(L=L_m\), and
\[
                  \frac KL>\frac{e}{e-1}-\frac\eps4.
\]
Put \(N=T=m!\), padding the horizontal candidate block with unused
positions \(m+1,\ldots,N\), setting \(z_j=0\) on these padded rows.
Form a strict-CSP seed with a row variable
\(h\) and a column variable \(v\), each of weight \(1/2\), common alphabet
\(2^{[N]}\), and cost \(c(A)=|A|\).  For a segment \(S_{j,i}\), let
\[
 I^{\rm H}_{j,i}=\{j\},\qquad
 I^{\rm V}_{j,i}
 =\{(i-1)T/j+1,\ldots,iT/j\},
\]
and impose on the ordered scope \((h,v)\) the rule
\[
 \cR_{j,i}=\{(A,B):
 A\cap I^{\rm H}_{j,i}\ne\varnothing\ \text{or}\
 B\cap I^{\rm V}_{j,i}\ne\varnothing\}.
\]
Keep one constraint for every segment.  The integral seed optimum is
\(K/2\), and \(\Top=[N]\) is absorbing.

Apply the ordered sampler \cref{eq:ordered-sampler} separately to the row
weights \(z_j\) and the column weights \(P/T\).  For each constraint the
two trace masses sum to at least one.  The coupled shifts from
\cref{eq:square-rotation} therefore give a rational feasible local
distribution with the prescribed marginals.  Independently replacing
each coordinate by \(\Top\) with probability \(\rho\) connects its support.
The perturbed LP value is
\[
                         \ell_\rho=(1-\rho)L/2+\rho N.
\]
Apply \cref{thm:transfer}, followed by the Cartesian cloning lemma
\cref{lem:cloning}.  Every resulting variable retains either row type or
column type, and every constraint retains its singleton horizontal trace
and consecutive vertical trace.  For fixed \(\eps\), the seed size
\(N=m!\), the alphabet, the Unique Games word length, and the tolerances
are constants; \cref{lem:cloning} therefore has polynomial output size.

We now realize the unweighted CSP geometrically.  Give every row-type clone
\(u\) its own \(N\) admissible horizontal rows \(y=Y_u+a\), \(a\in[N]\),
choosing the bases \(Y_u\) as consecutive multiples of \(N+1\).
Give every column-type clone \(v\) its own \(N\) admissible columns
\(x=X_v+b\), \(b\in[N]\), again using consecutive multiples of \(N+1\)
for the bases \(X_v\).  Its private x-coordinate interval is
\([X_v+3/4,X_v+N+1/4]\); these intervals are pairwise disjoint and contain
every expanded projection realized in their respective blocks.
Suppose a cloned
constraint on \((u,v)\) inherits the seed segment \(S_{j,i}\), and write
\(\ell=(i-1)T/j+1\) and \(r=iT/j\).  Create the horizontal segment
\[
 [X_v+\ell-\tfrac14,\ X_v+r+\tfrac14]\times\{Y_u+j\}.
\]
It is stabbed by exactly the row indexed by \(j\) among the candidates of
\(u\), and by exactly the columns indexed by
\(\ell,\ldots,r\) among the candidates of \(v\).  The private zones exclude
all candidates belonging to other variables.  Thus CSP labels and
finite-candidate stabbing sets are in a feasibility- and cost-preserving
bijection, up to the common normalization of clone weights.  Coincident
segments may be removed.  The coordinate magnitudes are \(O(M(N+1))\),
where \(M\) is the total number of clones, so their encoding is polynomial.

This realization also preserves the optimum when arbitrary lines are
allowed.  Every useful horizontal line is a supporting candidate row.
A useful vertical line lies in one private zone.  If it stabs segments
whose vertical traces are \([b_f,d_f]\cap\mathbb Z\), set
\(B=\max_f b_f\) and \(D=\min_f d_f\).  Their geometric intervals have a
common point, so \(B-1/4\le D+1/4\).  Since \(B,D\) are integers, \(B\le D\),
and the candidate column indexed by \(B\) stabs all of them.  Snapping every
vertical line in this way never increases the cover size.  For a
fractional cover, move the mass to these columns, combine coincident
masses, and cap each total at one.  Feasibility is preserved and cost
does not increase.  Since the supplied candidates remain admissible, both
optima agree in the two models.  Every realized segment
has length at most \(N-1/2\), a constant depending only on \(\eps\).

Finally choose \(\rho\), the transfer error \(\delta\), and the cloning
tolerance \(\eta\) sufficiently small that
\[
 \frac{K/2-\delta}{\ell_\rho+\delta+\eta}
 >\frac{e}{e-1}-\eps.
\]
The transfer theorem then proves the claimed hardness in both models.
\end{proof}

\end{document}